\documentclass{acm_proc_article-sp}

\usepackage{epsf}
\usepackage{color}
\usepackage{epsfig}
\usepackage{graphicx}
\usepackage{subfigure}
\usepackage{amsmath,amssymb}
\usepackage{algorithm}
\usepackage[noend]{algorithmic}
\usepackage{multirow}
\usepackage{xspace}
\usepackage{url}
\usepackage[normalem]{ulem}

\newcommand{\eat}[1]{}

\newtheorem{mydef}{\textbf{DEFINITION}}

\newtheorem{mytheo}{\textbf{THEOREM}}
\newtheorem{mylemma}{\textbf{LEMMA}}

\newlength{\figwidth}
\newlength{\figthree}
\newlength{\figfour}
\newlength{\figfours}
\begin{document}

\title{DataGrinder: Fast, Accurate, Fully non-Parametric Classification Approach Using 2D Convex Hulls}

\author{
Mohammad Khabbaz\\
       \affaddr{Personal Business Development}\\
       \email{mohammmad@gmail.com}
}

\maketitle

\begin{abstract}
It has been a long time, since data mining technologies have made their ways to the field of data management.
Classification is one of the most important data mining tasks for label prediction, categorization
of objects into groups, advertisement and data management. In this paper, we focus on the standard
classification problem which is predicting unknown labels in Euclidean space. Most efforts in Machine
Learning communities are devoted to methods that use probabilistic algorithms which are heavy on
Calculus and Linear Algebra. Most of these techniques have scalability issues for big data, and are hardly parallelizable
if they are to maintain their high accuracies in their standard form. Sampling is a new direction for improving scalability, using many small parallel classifiers. In this paper,
rather than conventional sampling methods, we focus on a discrete classification algorithm with $O(n)$ expected running time. Our approach performs a similar task as sampling methods. However, we use column-wise sampling of
data, rather than the row-wise sampling used in the literature. In either case, our algorithm is completely
deterministic. Our algorithm, proposes a way of combining 2D convex hulls in order to achieve high classification
accuracy as well as scalability in the same time. First, we thoroughly describe and prove our $O(n)$ algorithm for
finding the convex hull of a point set in 2D. Then, we show with experiments our classifier model built based on this idea
is very competitive compared with existing sophisticated classification algorithms included in commercial
statistical applications such as MATLAB.
\end{abstract}

\section{Introduction}
\label{sec:intro}
Data mining topics such as classification~\cite{AutoAtt:11, SVMInd:11, FineClassification:13, GAIA:10, Raymond:12, SMCC:12, Intrusion:10, DiscriminantClassifier:07}, clustering~\cite{BIRCH:96, DDFactors:11, DCInference:12, ClusterForest:13, SMCC:12, AdvancedClust:11, ArrayStore:11, LatentOlap:11, SparseGraph:11}, frequent pattern mining~\cite{ARM2:00, ARM1:94, MaxFreq:13}, frequent sub-structure mining~\cite{MaxFreq:13, SMCC:12}, regression~\cite{Bilal:11, Ashraf:11}, data cleaning~\cite{ERACER:10, AdvancedClust:11, AutoAtt:11}, ranking~\cite{PageRank:11}, data warehousing~\cite{LatentOlap:11}, recommender systems~\cite{TopRecs:11, Recsplorer:10, Package:12}, bio-informatics~\cite{Raymond:12}, outlier detection~\cite{BSkyTree:10, IOSky:13, efficientSkyline:11}, nearest neighbors~\cite{FuzzyNN:10, NN:12}
and social networks~\cite{SN1,SN2}, have been widely discussed in data management and prediction.
There has been plenty of work on classification as one of the main techniques
for supervised learning. Figure~\ref{fig:classificationexp}, shows a small example where
we have sets of points in a plane, each of which belonging to one category, demonstrated by different
shapes. A classifier model, is given data vectors in 2 dimensional space (2D), with labels (i.e. training), and is
expected to predict, and assign new objects with missing labels to their correct categories(i.e. testing).

\begin{figure}[t]
  \centering
  \includegraphics[scale=0.4, trim = 20mm 0mm 0mm 0mm]{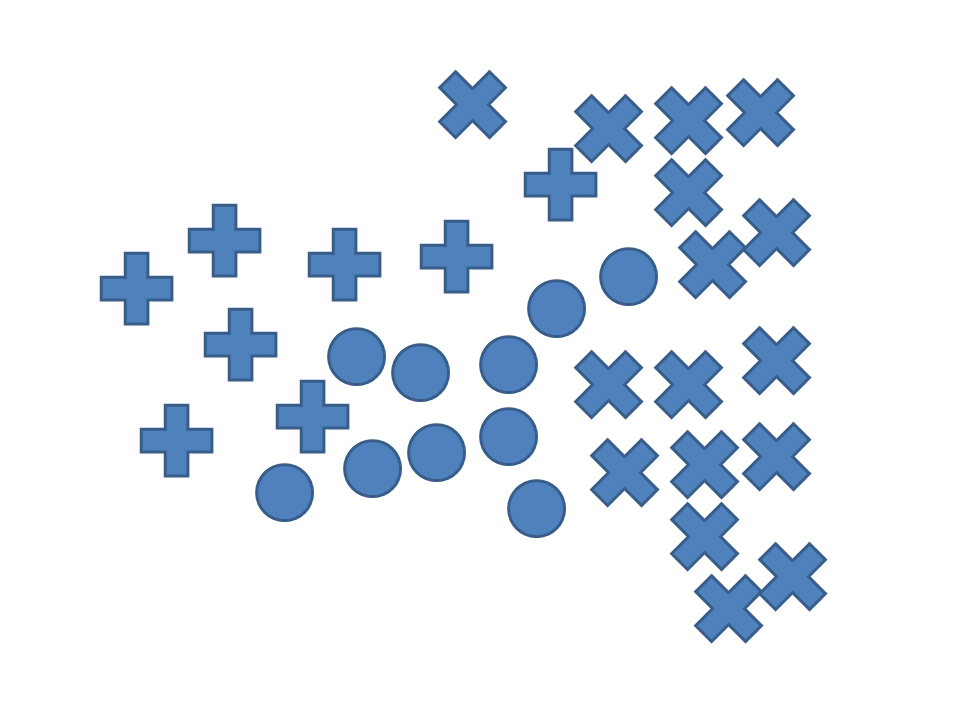}
  \caption{Classification Example with 3 Classes}
  \label{fig:classificationexp}
\end{figure}

There exist a variety of classification algorithms in machine learning and data mining.
Most popular classifiers are known as discriminant classifiers. Discriminant classifiers
aim at statistical or probabilistic modeling in order to find an objective function. Then,
optimization and numerical methods are used in order to find optimal parameter values. Having found
these optimal values, we can find the decision boundaries that divide the space into regions
that separate objects from different categories. It is often the case that data is not linearly
separable. This leads to misclassification errors most of the times. In order to minimize misclassification
error, people use methods such as regularization, kernel transformations, feature extraction and feature selection~\cite{BishopML}. Examples of discriminant classifiers include Support Vector Machines and Logistic Regression~\cite{SVMInd:11, DiscriminantClassifier:07}.

Other types of classifiers are Decision Trees, Rule-based~\cite{RuleClass1, RuleClass2, RuleClass3} methods and Nearest Neighbor methods. Most of
these methods have practical shortcomings. Discriminant classifiers need to optimize their objective
function and this may not be feasible in reasonable time for big data. Besides, it is a challenge
to find straight forward parallel implementations of these optimization algorithms. In web-based
scenarios, data changes very frequently~\cite{SMCC:12}. This requires either algorithms with highly scalable training
 phase, or models we can sequentially update. Although time always plays a key role and sometimes
 sequential update may not be optimal~\cite{RecSys}. Decision Trees and
Rule-based classifiers also suffer from the same shortcoming in practical scenarios. In many cases,
the theoretical problem defined to solve the classification problem is NP-hard. Nearest Neighbor methods
are efficiently applicable if data is stored in data structures such as kd-trees for nearest neighbor search. Despite their efficiency
in execution, they lack accuracy even for slightly challenging inputs. We demonstrate this with experiments in
Section~\ref{sec:expn}, and briefly explain how each classification algorithm works.

In this paper, rather than solving optimization problem, we use \textbf{\emph{Computational Geometry}}, in order to
build an accurate classifier. We use 2D convex hulls, using all possible 2 dimensional
projections (i.e. all possible pairs of columns regardless of order). Figure~\ref{fig:convex1}, shows an example
of the convex hull of a point set $P$. In order to build classifiers, we project the input dataset with $d$ dimensions to all possible $\binom{d}{2}$ planes. In each plane, having partitioned the training data into different classes, we find the 2D convex hull for each class (\emph{Select-Project-ConvexHull}). This results in $C \times \binom{d}{2}$ convex hulls, where $C$ is
the number of classes. Given a new testing instance with
$d$ feature values, we check for all existing $C \times \binom{d}{2}$ convex hulls, whether they contain the corresponding $2$ dimensional projection($\pi$).
We find the class $c_{max}$, that scores highest (i.e. its boundaries contain the point in more
2D projections) and assign the class label.
Since $d$ is typically a small constant in practice, we are not worried about the testing time.
Besides, using parallelization, testing time is negligible. We also propose a filtering approach to choose only the most discriminant features in Section~\ref{sec:expn},
that results in accuracy improvements as well. We explain our classification algorithm in more detail in Section~\ref{sec:classificationAlgo}, after providing the necessary computational geometry
 background. We make the following contributions in this paper:

\begin{enumerate}
\item We explain the Convex Hull problem from Computational Geometry~\cite{ComputationalGeometryBook}.
We provide algorithmic background in terms of the running time, and propose
an algorithm with $O(n)$ expected running time. We also prove its correctness.
Besides, we calculate the \emph{"constant"} through probabilistic analysis, and our
experiments show our calculated constant is reliable for different sizes
of data. Database community has shown tremendous interest in solving problems
formulated similar to convex hulls such as designing algorithms for finding Skylines~\cite{BSkyTree:10, IOSky:13, efficientSkyline:11}.
\item We propose and explain our classification algorithm, DataGrinder(DGR),
using 2D convex hulls. We also propose tricks for tuning the classifier by
filtering weak features that results in considerable accuracy improvement, in the case of one dataset.
\item We propose parallel algorithms for implementing DataGrinder at different
levels including data partitioning as well as parallel convex hull algorithms,
using the divide and conquer method.
\item We propose a method for random data generation and testing classifiers.
Our proposed testing methodology controls the hardness of classification using two
parameters. We conduct a comprehensive set of experiments on randomly generated
and real datasets. Our experiments show DataGrinder is competitive against the most widely known
commercial classifiers in accuracy, while being extremely scalable.
\end{enumerate}

\section{Convex Hull Background}
\label{sec:convexBackground}
\begin{figure}[t]
\subfigure[Convex Hull of $P$]{
\includegraphics[scale=0.3]{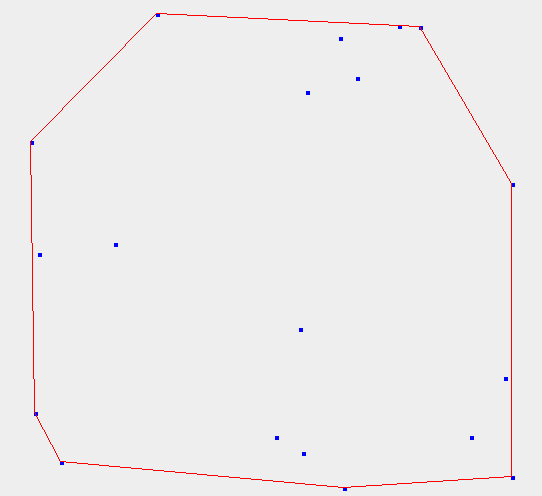}
\label{fig:convex1}
}
\subfigure[Convex Hull and Line Segments]{
\includegraphics[scale=0.5]{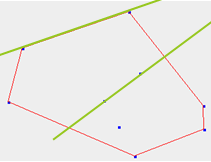}
\label{fig:convex2}
}
\caption{Convex Hull of a point set $P$. The set of all line segments that have
every other point only in one of their half spaces is a polygon that is called the convex hull $CH(P)$.}
\end{figure}

Convex Hull of a point set in 2D, $CH(P)$, is a set of points such that every
point in $P$ can be computed as a positive linear combination (all the weights are positive), of the points in $CH(P)$.
For this reason, it is important in applications where we are interested in
finding mixtures using some baseline prototype vectors. In a sense,
convex hull of a point set is a small subset of the points that wraps
around a point set, and can represent any point in the point set with
a positive linear combination. We can represent convex hull as a polygon,
that contains every other point in $P$, within its boundaries. This polygon
that wraps around all points can be extremely useful for applications in Machine Learning
when we want to define the boundaries of a group of points (class). It also
has the base ingredients to represent any point in that class. For instance,
we can use convex hull along with radial basis or any other sort of function
in order to construct a kernel and represent every point in a new feature space.

Moreover, in Computational Geometry, problems such as finding \emph{half space intersection}
can be reduced to finding convex hull and this highlights the importance of exploring more computationally
efficient algorithms. In many real life applications we deal with datasets with thousands
or millions of data points, and existing $O(nlog(n))$ algorithms fail to find convex hull in a timely manner.
Other than problems we can directly model with convex hull, there are many other domains such as
web mining where we deal with large graphs.  We can also use properties of these domains such as
link structure in order to define entities such as web pages in a multidimensional Euclidean space, and
use convex hull for modeling~\cite{SMCC:12}.

We only focus on the 2D case but our heuristics and ideas are generalizable to
higher dimensions. We use 2D for simplicity because it is more intuitive for problem solving
and leave generalized version to our future work. Moreover, in our present application of convex hulls
(i.e. classification), we are seeking data boundaries as tight as possible while still maintaining
properties of binary feature correlations. This, kind of resembles entries of a covariance matrix,
in Multivariate Gaussian Distributions that can be used for Principle Component Analysis as well.

It is provable that the best possible worse case running time for this problem is $O(nlog(n))$, since
sorting can be reduced to the convex hull problem~\cite{ComputationalGeometryBook}. In fact, the most
efficient classic convex hull algorithms use sorting. First, we sort all the points in a dataset according to one coordinate. Then,
using a left to right scan of the sorted list, we iterate over other points and remove
any points that do not belong to the convex hull, in linear time. In order to do so, they use geometric
properties of the points on convex hull and line segments between them. Figure~\ref{fig:convex1} shows
a point set along with a polygon that wraps around it. If we extend each line segment
in both ends, we obtain a line such that every other point in $P$, is located on one side (i.e. half space).

We devise an algorithm, that despite its $O(n^2)$ worst case running time,
achieves $O(n)$ expected running time if $P$ is distributed uniformly,
and dimensions are independent. We also prove for independent \emph{Normal} distributions. Previous work in Computational Geometry also approves the possibility
of $O(n)$ expected running time, if the algorithm is designed within the given framework~\cite{Convex:81}.
Here, we thoroughly describe the algorithm and provide pseudo code as well as average case
analysis for computing the constant. Regardless of the data instance, we can always devise strategies
to avoid the worst case through smart query optimization, and use of empirical algorithms.

Our 2D convex hull algorithm avoids paying the initial $nlogn$ sorting time. Instead, in every iteration
our new algorithm finds the next minimum of the list in $O(|candidates|)$, and using the new point, it
uses a heuristic to remove other points from the candidate set, that do not qualify to be on convex hull.
This is what we refer to as \emph{Candidate Elimination} process. Once we process all the candidates and
remain with an empty \textbf{\emph{candidate set}}, we have found the convex hull. Our theoretical analysis
as well as our quantitative experimental results, suggest that repeating this process results in $O(n)$ expected
running time, for finding 2D convex hull. \textbf{This iterative candidate elimination process enables us to find the convex
hull of up to $1000000$ points in less than $20$ seconds while the existing classic algorithm fails to terminate
in in a timely manner (after 8 hours)}. It is worth highlighting again that although the classic algorithm
 has a better worst case running time, it fails in practice. In the rest of this section, we formally define the convex hull problem and
discuss naive and classic solutions. In the subsequent subsections, we discuss a new algorithm based on candidate elimination,
and discuss its expected running time. Eventually, we show with experiments that the improvement achieved using
this pruning heuristic is indeed considerable, and indeed it results in linear expected running time.

\subsection{\textbf{Convex Hull Problem Definition}}
\label{sec:probdef}
Convex Hull of a point set $P = \{p_1, p_2, ..., p_n\}$, is best defined
intuitively as a polygon that wraps around all the points in $P$. We can
formally define this polygon as follows.

\begin{mydef}
\textbf{\emph{Convex Hull}} of a point set, $CH(P)$, is the set of \textbf{\emph{all}}
line segments, $pq$, between every two pair of points from $P$, such that
every other point is located on one side of $pq$. We can also use negative or positive, in
order to refer to these two \emph{"half spaces"}. In other words, every other point either
belongs to the negative half space, or to the positive half space.
\end{mydef}

Naive algorithm for finding $CH(P)$ is as follows:

\begin{itemize}
  \item Produce every pair of points $p_i$ and $p_j$: $O(n^2)$
  \item Find the line segment between $p_i$ and $p_j$ in \emph{constant time}.
  \item Check if every other point belongs to either negative or positive half space.
  If yes, add the line segment to $CH(P)$ otherwise discard: $O(n)$.
\end{itemize}

Figure~\ref{fig:convex2}, shows examples of both types of line segments.
Overall running time of the naive algorithm is $O(n^3)$, since
it scans $P$ once for every pair of points. This results in a process
that takes minimal usage of geometric properties and is extremely inefficient.
Using geometric properties, we can aim at designing a more targeted
process. Next, we describe $O(nlog(n))$ algorithm that first sorts
all the points by their $x$-coordinate.

\subsection{\textbf{Background of Algorithms ($nlogn$ algorithm)}}
\label{sec:classicAlgo}

Rather than arbitrarily exploring the search space, first
we sort the point set based on one coordinate (typically $x$).
Points in $P$, start from $X_{min}$ and end at $X_{max}$ after
sorting.

\begin{figure}[t]
  \centering
  \includegraphics[scale=0.4, trim = 0mm 0mm 0mm 0mm]{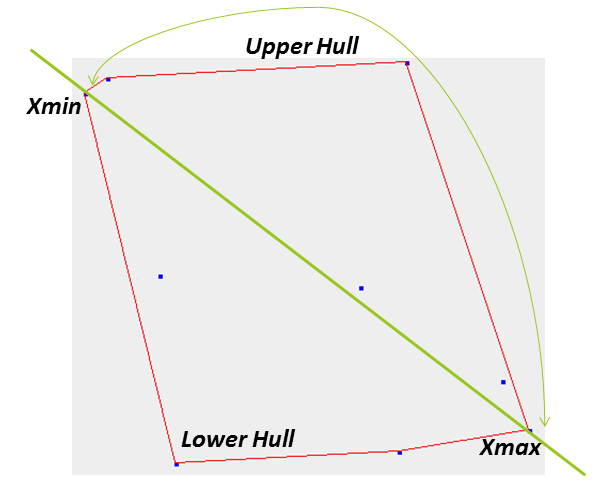}
  \caption{Upper and Lower Hulls between $X_{min}$ and $X_{max}$.}
  \label{fig:UpperLowerHulls}
\end{figure}

Figure~\ref{fig:UpperLowerHulls}, splits the convex hull into two
parts, both from $X_{min}$ to $X_{max}$. We use Upper Hull ($UH$), in order
to refer to the part above the line segment between $X_{min}$ and $X_{max}$;
we use Lower Hull ($LH$), in order to refer to the lower part. \emph{Classic algorithm}
for finding convex hull invokes $FindUpperHull(P)$ and $FindLowerHull(P)$ functions,
in order to find the convex hull of $P$, each in $O(n)$. Therefore, the total execution time
is $O(nlog(n))$. Finding upper and lower hulls separately are two symmetric procedures
with respect to each other. Here, we only present for upper hull. Algorithm~\ref{algo:findUpper},
computes the upper hull of the sorted point set by scanning from $X_{min}$ to $X_{max}$. It is
intuitive that we visit all the points in $UH(P)$ in sequence, once we do scanning from
left to right, although with the rest of the points in between. The idea is to: 1) perform
this scanning; 2) identify and maintain points that belong to $UH(P)$, and 3) discard all the other
points. We change $i$ from $1$ to $n$, and start
the $i^{th}$ iteration having computed the correct upper hull of the points $\{p_1 ... p_{i-1}\}$. We
add $p_{i}$ to $UH(P)$, because we know it belongs to the upper hull of $\{p_1 ... p_i\}$,
with the largest $x-$coordinate value so far. We read $UH(P)$ in reverse order and remove any
points that do not belong to the convex hull of $\{p_1 ... p_i\}$, until we stop.

\begin{algorithm}
\begin{algorithmic}[1]
\caption{FindUpperHull$(P)$}
\label{algo:findUpper}
\REQUIRE Point set $P$,
sorted by $x-$coordinate
\ENSURE Upper hull of $P$, $UH(P)$
\STATE $UH = $ initialize empty
\FOR{$i=1$ to $n$}
\STATE $UH.append(p_i)$
\STATE $\ell = i$
\WHILE{$(\ell > 2) $\&\&$ (!UHCheck(p_{\ell-2}, p_{\ell-1}, p_\ell))$}
\STATE remove $p_{\ell-1}$ from $UH$
\STATE $\ell = \ell - 1$
\ENDWHILE
\ENDFOR
\STATE \textbf{return} $UH$
\end{algorithmic}
\end{algorithm}

After appending $p_i$ to $UH$, we check for the last $3$ points
in $UH$, if they belong to the correct convex hull or not. In order
to do so, $UHCheck(p_{\ell-2}, p_{\ell-1}, p_{\ell})$ returns true if $p_{\ell-1}$ is above
the line segment from $p_{\ell-2}$ to $p_{\ell}$. This means $p_{\ell-1}$ belongs to $UH_{i}$.
Otherwise, it is removed and we repeat this process until $UHCheck$ returns \textbf{true}, and obtain
the correct upper hull of $p_1$ to $p_i$. Equation~\ref{eq:checkUp}, computes a \emph{sign} variable.
If sign is a non-negative number, $UHCheck$ returns true.
%		int sign = (point.y - lineSegmentBegin.y) * (lineSegmentEnd.x - lineSegmentBegin.x) -
%				(point.x - lineSegmentBegin.x) * (lineSegmentEnd.y - lineSegmentBegin.y);

%sign = (p_{\ell-1}.y - p_{\ell-2}.y)(p_{\ell}.x-p_{\ell-2}.x) - (p_{\ell-1}.x - p_{\ell-2}.x)(p_{\ell}.y-p_{\ell-2}.y)

\begin{equation}
\label{eq:checkUp}
\begin{array}[b]{r}
sign = (p_{\ell-1}.y - p_{\ell-2}.y)(p_{\ell}.x-p_{\ell-2}.x) - \\
(p_{\ell-1}.x - p_{\ell-2}.x)(p_{\ell}.y-p_{\ell-2}.y)
\end{array}
\end{equation}

Figure~\ref{fig:classicAlgoEx}, shows a snapshot during the execution, where two
middle points need to be removed after adding $p_5$.
It also shows the correct upper hull after we exit the while loop. We exit the while
loop when for the first time we find a middle point which passes the convex test (Equation~\ref{eq:checkUp}).
When this happens, it is guaranteed that $UH_5(P)$ is convex since the last step is convex and also
we know the rest is constructed convex starting from $p_1 = X_{min}$. We exit the while loop also when only
$2$ points are left, in which case there is no middle point and $UH_i(P)$ is always convex.

\begin{figure}[t]
  \centering
  \includegraphics[scale=0.5, trim = 0mm 0mm 0mm 0mm]{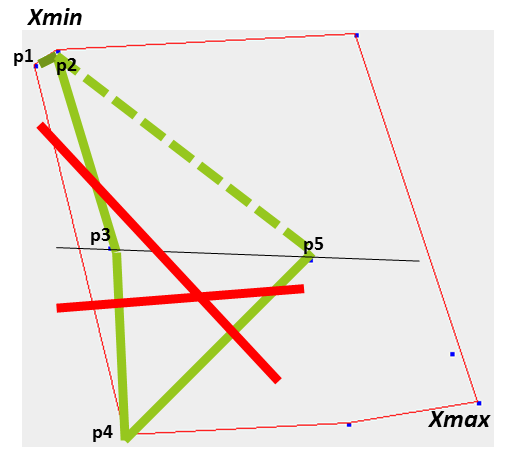}
  \caption{A snapshot of running the classic algorithm for finding upper hull. Convex Hull of $p_1 ... p_4$ is
  constructed by adding the points one by one without removing any points. When $p_5$ is added, $p_4$ needs
  to be removed because upper hull of $p_1 ... p_5$ is not convex anymore. Moving backwards iteratively, FindUpperHull (Algorithm~\ref{algo:findUpper}), removes $p_4$ and $p_3$ and we remain with $UH_5(P) = \{p_1, p_2, p_5\}$.}
  \label{fig:classicAlgoEx}
\end{figure}

It is worth noting, it can happen that two points appear in sorted $P$ next
to each other with the same value for $x-$coordinate. In this situation, we order
all the points with the same $x-$value based on their $y-$coordinate to preserve
the correctness of Algorithm~\ref{algo:findUpper}. If two points
are exactly the same, one can be removed without hurting the correctness of the convex hull
algorithm.

\section{Finding Convex Hull by Candidate Elimination}
\label{sec:newAlgo}

Classic convex hull algorithm presented so far needs to perform an
initial sorting with cost $O(nlog(n))$. We know we can not do better
in the worst case for finding convex hull. Despite $O(nlog(n))$ worst case
running time, in many cases we may be able to use heuristics in order
to make the problem size smaller and achieve better Expected running time.
In this section, we describe a process called \emph{"Candidate Elimination"},
that we use, instead of sorting. We use candidate elimination along with existing
\emph{FindUpperHull} procedure, in order to solve the problem. The idea is
to avoid sorting, maintain candidate lists of points for different parts
of the convex hull, and find the next \emph{minimum} value from a smaller
candidate list, rather than paying $O(nlog(n))$ for sorting in the beginning.

\begin{figure}[t]
  \centering
  \includegraphics[scale=0.4, trim = 0mm 0mm 0mm 0mm]{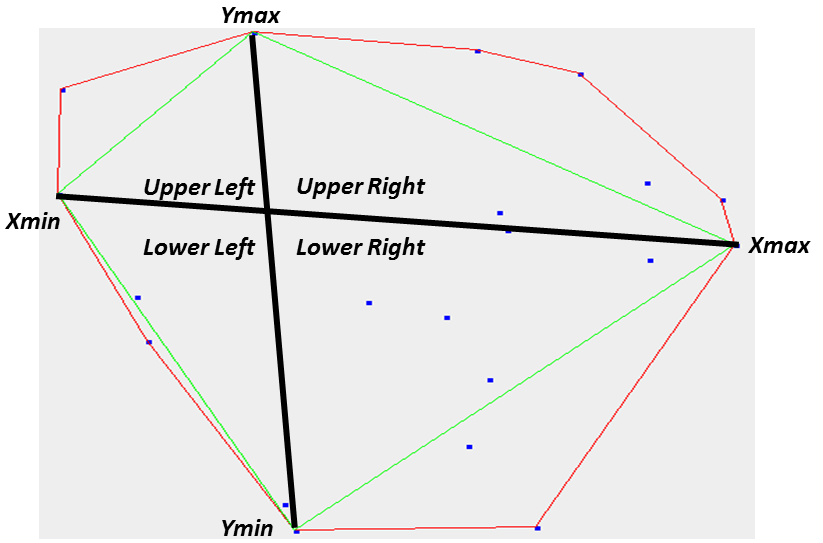}
  \caption{Dividing the plane into 4 quarters using $X_{min}$, $X_{max}$, $Y_{min}$ and $Y_{max}$.}
  \label{fig:CandidateElimination1}
\end{figure}

Figure~\ref{fig:CandidateElimination1}, divides the plane as well as the convex hull of the
point set into $4$ quarters, using minimum and maximum $x$ and $y$ values in the point set.
We use $UpperLeftHull$, $UpperRightHull$, $LowerLeftHull$ and $LowerRightHull$, in order to
refer to these $4$ quarters.

\begin{mylemma}
\label{lemma:candidate}
All of the points on $UpperLeftHull$ are on or above the line from $X_{min}$ to $Y_{max}$.
\end{mylemma}
\begin{proof}
We know the upper left hull starts at $X_{min}$ and ends at $Y_{max}$. We also know that
the upper left hull is convex. Therefore, none of the points on it can be below
the line.
\end{proof}

Lemma~\ref{lemma:candidate}, provides an opportunity for candidate elimination in the beginning.
We can draw a line from $X_{min}$ to $Y_{max}$, and remove any points below the line, to obtain
a list of $UpperLeftHull$ candidates. Using symmetry, we can find a candidate
list for $UpperRightHull$ by choosing all the points above the line that goes through
$Y_{max}$ and $X_{max}$. Lower left candidates are those on or below the line from $X_{min}$
to $Y_{min}$, and lower right candidates are on or below the line from $Y_{min}$ to $X_{max}$.
Finding minimum and maximum $x$ and $y$ coordinate values can be done in $O(n)$. Therefore,
by paying $O(n)$, we can discard many points and continue with smaller input size and this
obviously can considerably improve the performance. We use \emph{Candidate Elimination}, to
refer to this process that makes more targeted use of both $x$ and $y$ coordinates. Figure~\ref{fig:CandidateArea},
shows a minimal box that contains all of the points in $P$, using $X_{min}$, $X_{max}$, $Y_{min}$ and $Y_{max}$.
Inside this box, we separate $4$ triangles in $4$ corners. These are the only areas where
convex hull candidates can appear. We use \emph{\textbf{Candidate Area}} in order to refer to any
area inside the box, where convex hull candidates can appear. In Figure~\ref{fig:CandidateArea},
four triangles form the candidate area.

\begin{figure}[t]
  \centering
  \includegraphics[scale=0.4, trim = 0mm 0mm 0mm 0mm]{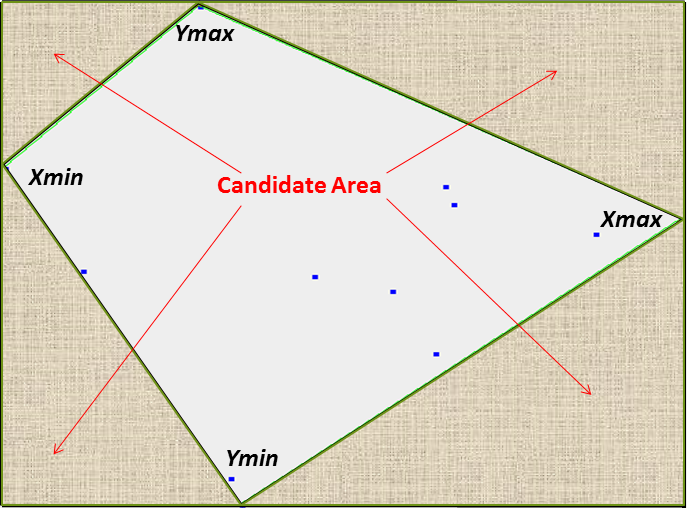}
  \caption{Box area vs. Candidate Area.}
  \label{fig:CandidateArea}
\end{figure}

\begin{mylemma}
\label{lemma:Area1}
The expected number of candidates after the first candidate elimination is $n/2$.
\end{mylemma}
\begin{proof}
We assume points are distributed uniformly in the plane. We also assume that $x$ and $y$
coordinates are uniform and independent. Given these assumptions, we define $z_i$ to be
a random variable. We assign $z_i = 1$, if the $i^{th}$ point in $P$ is in the candidate area.
We know $P(z_i = 1) = (CandidateArea/BoxArea) = 1/2$; therefore, $E(z_i) = 1/2$. There are
$n$ such points in the dataset, and we can use $E(Z) = \sum_{i=1}^{n} E(z_i)$ while $Z$ is
a random variable that takes values in $\{1 ... n\}$, that indicates the number of candidate
points all together after the first candidate elimination. Expected value of $Z$ is $n$
times expected value of $z_i$, equal to $n/2$, using linearity of expected value.
\end{proof}

\subsection{Convex Hull Algorithm}
\label{sec:algo}

The first candidate elimination step reduces
the expected number of candidates to half. Although this is a good heuristic,
we still need to eliminate more candidates, and find the correct convex hull. As
described earlier, we do this in $4$ smaller steps for $UpperLeftHull$, $UpperRightHull$,
$LowerLeftHull$, and $LowerRightHull$, separately. Here, we only describe the process for
$UpperLeftHull$, and we know the rest is symmetric for the three other quarters of the convex hull.
Algorithm~\ref{algo:findUpperLeft}, takes as input the list of upper left
candidates after the initial candidate elimination, that are on or above the line from $X_{min}$ to $Y_{max}$.
Please note, that the list is not sorted by $x-$coordinate anymore.
The idea is to avoid sorting the candidate list. Instead, we keep finding the
next smallest $x$, $NextX$, and repeat candidate elimination using $NextX$.
The justification behind replacing sorting with this operation, is the fact
that candidate list keeps getting smaller and smaller after performing
candidate eliminations in sequence. This makes the cost of finding
the next minimum negligible, even for large $n$.

\begin{algorithm}
\begin{algorithmic}[1]
\caption{FindUpperLeftHull$(ULCandidates)$}
\label{algo:findUpperLeft}
\REQUIRE $ULCandidates$,
list of candidates for upper left hull
\ENSURE Upper left hull of $P$, $ULH(P)$
\STATE $ULH = $ initialize empty
\WHILE{$ULCandidates.size > 0$ }
\STATE $NextX = removeLeftMostPoint(ULCandidates)$
\STATE $eliminateCandidates(ULCandidates, NextX, Y_{max})$
\STATE $ULH.append(NextX)$
\STATE $\ell = ULH.size$
\WHILE{$(\ell > 2) $\&\&$ (!UHCheck(p_{\ell-2}, p_{\ell-1}, p_\ell))$}
\STATE remove $p_{\ell-1}$ from $ULH$
\STATE $\ell = \ell - 1$
\ENDWHILE
\ENDWHILE
\STATE \textbf{return} $ULH$
\end{algorithmic}
\end{algorithm}

Rather than reading the next point from sorted $P$, in order to find upper left
hull, Algorithm~\ref{algo:findUpperLeft}, finds $NextX$ in line $3$ and removes
it from the list of upper left candidates. We pay $O(|ULCandidates|)$ cost to find
$NextX$. In line $4$, $eliminateCandidates$ repeats the same candidate elimination task using
$NextX$. In order to do so, we draw a line from $NextX$ to $Y_{max}$, and remove any
candidates below the line. In the rest of Algorithm~\ref{algo:findUpperLeft},
we pretend $NextX$ is read from a sorted list and repeat the same process
in order to fix $UpperHull$ that Algorithm~\ref{algo:findUpper} does, already
presented in Section~\ref{sec:classicAlgo}.

\section{Running Time Analysis}
\label{sec:runningTime}

\subsection{Worst Case Running Time}
\label{sec:worstTime}

There are three main steps in each iteration of finding convex hull by \emph{candidate elimination}:

\begin{itemize}
  \item Finding $NextX$, overall $O(|candidates|)$
  \item Candidate elimination, $O(|candidates|)$
  \item Fixing upper hull, $C$ (constant)
\end{itemize}

It is possible in the worst case, that all of the points in $P$ belong to the convex hull.
In this case, candidate elimination results in removing no candidates and repeating a $O(n)$
process $n$ times, resulting in $O(n^2)$ \emph{worst case} running time. For the current
classification problem, worst case scenario rarely happens.

\subsection{Expected Running Time}
\label{sec:expectedTime}

Since there are $4$ quarters and the expected number of candidates is $n/2$
after the initial candidate elimination, there is an expected number of $n/8$
candidates in each triangle. It is worth noting, we can use the product of
expected values of two random variables as the expected value of their product, because
all the random variables are independent~\footnote{Points are independently drawn from the uniform distribution.}.
This, is a natural assumption, used widely in Machine Learning~\cite{BishopML}.
We define $\alpha_0 = 1/8$ to be the elimination ratio, indicating the expected cost of finding
$NextX$, after the initial candidate elimination in each quarter. We present using $LowerRightHull$, to have more variety
in our examples. Subsequently, we can define, $0 < \alpha_1 < 1$, as elimination ratio in iteration
$1$ and, $0 < \alpha_2 < 1$, as elimination ratio in iteration $2$. The expected size of $LRCandidates$
after iteration $2$ is $(\alpha_0\alpha_1\alpha_2) \times n$.

\begin{mylemma}
Expected running time of finding the convex hull of the lower right quarter is $n/8 \times (\sum_{i=1}^{n/8}(\prod_{j=1}^{i}\alpha_j))$.
\end{mylemma}
\begin{proof}
We know $\alpha_0$ is the initial elimination ratio that reduces the number of lower
right candidates to $n/8$. Therefore, this is the expected size, we start with. In each
iteration, we pay the cost $O(|LRCandidates|)$. The number of $LRCandidates$ after iteration
$2$, is $\alpha_1\alpha_2 \times n/8$. Similarly, the number of candidates after the $i^{th}$ iteration
is $\prod_{j=1}^{i}\alpha_j$. Therefore, we pay $n/8 \times \prod_{j=1}^{i}\alpha_j$ cost, which
is the expected size of $LRCandidates$. Adding up for a maximum of $n/8$ iterations we get $(\sum_{i=1}^{n/8}(\prod_{j=1}^{i}\alpha_j))$, the total expected cost of finding $LowerRightHull$.
Although we write the sum for $n/8$ iterations, it is quite likely that in the end the expected
cost is $0$ or close to $0$. This is because an exponentially smaller coefficient is
multiplied by $n/8$. This is because all $\alpha_1 ... \alpha_{n/8}$ are smaller than $1$.
\end{proof}

\begin{mylemma}
\label{lemma:candidate2}
$\alpha$ is a decreasing function that approaches $3/4$.
\end{mylemma}
\begin{proof}
Suppose at some iteration we have found $NextX$ and we perform
candidate elimination. Figure~\ref{fig:CandidateArea2}, compares
candidate area to eliminated area. Candidate area is shown below the
line from $NextX$ to $X_{max}$. Eliminated area is a triangle with $area = ac/2$.
Total area is $bc/2 + ac + ac/2$. Therefore, $\alpha = (bc/2 + ac)/(bc/2 + ac + ac/2) = (2a+b)/(3a+b)$.
As we get closer to $X_{max}$, $b$ gets closer to $a$ and the value of $\alpha$ decreases to $3a/4a = 3/4$.
\end{proof}

\begin{figure}[t]
  \centering
  \includegraphics[scale=0.4, trim = 0mm 0mm 0mm 0mm]{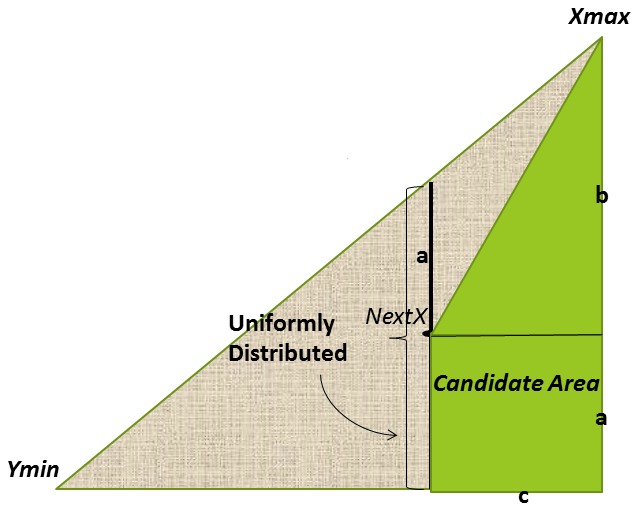}
  \caption{Candidate elimination after finding $NextX$}
  \label{fig:CandidateArea2}
\end{figure}

Using lemma~\ref{lemma:candidate2}, we know $(\prod_{j=1}^{i}\alpha_j)$ is
a product that decreases with $i$. Since $\alpha$ is smaller than $1$ all the time, and
$\alpha$ is a decreasing function. We can \emph{"assume"} $(\prod_{j=1}^{i}\alpha_j)$
exponentially decreases with $i$ and we can bound the expected running time using the sum of
a geometric series as follows: $n/8 \times \sum_{i=1}^{n/8}\bar{\alpha}^i$. Since $\bar{\alpha}$ is
a constant between $0$ and $1$, we know the sum of the geometric series is constant and so is
the expected running time. We know minimum value for $\alpha$ is $3/4$ and $\alpha < 1$. Using
average value of $3/4$ and $1$, we can approximate $\bar{\alpha} = 7/8$, resulting in $O(n)$ points
accessed during the execution of the convex hull algorithm for each corner. Finally, we can approximate $4n$
as the total number of points accessed during the execution for finding the convex hull of $4$ quarters. Next, we aim at calculating a \textbf{\emph{constant}} upper bound for the expected cost, in order to \textbf{\emph{prove}} the expected cost is linear, when convex hull is found by candidate elimination, instead of using $\bar{\alpha}$ which
is only raw approximation!

\begin{mytheo}
\label{theo:linearProof}
The expected value of $\alpha_1$ is constant $< 1$ and expected running time is bounded by
the sum of $\alpha_1$'s geometric series.
\end{mytheo}
\begin{proof}
In order to choose $NextX$, we need to draw a point from the uniform distribution
specified by the triangle in Figure~\ref{fig:CandidateArea2}. The three corners of the
triangle have these coordinates:$(Y_{min}.x, Y_{min}.y)$, $(X_{max}.x, Y_{min}.y)$,
\\ $(X_{max}.x, X_{max}.y)$. We are interested in finding the expected position of $NextX$
on $x-$axis. Since the distribution is uniform and we are interested in expected $NextX.x$,
we need to find the point on $x-$axis, such that if we split the triangle using a vertical
line, candidate areas \emph{inside the triangle} on both sides of the vertical line are equal.
We assume the perpendicular
sides of the triangle have equal expected length. One is equal to $L1 = X_{max}.x - Y_{min}.x$, and the
other equal to $L2 = X_{max}.y - Y_{min}.y$. $L1$ and $L2$ are two random variables. $E(L1/L2)$ depends
on the range of values of $x$ and $y$ coordinates in the point set. It is usually the case that these
coordinates are either in the same range or we can perform normalization and make expected values
of $L1$ and $L2$ both equal to a value $L$. Thus, without loss of generality we calculate $(L-c)^2/2$
as the triangle area on the left side of $NextX$. We also compute $(2L-c)\times c/2$, the area
inside the triangle on the right side of $NextX$. Therefore, we need to find the value of $c$ in terms of $L$ in the following equation: $$ L^2 + c^2 - 2Lc = 2Lc - c^2$$.

We get $c = L/(2+\sqrt{2})= L/3.4 \approx 0.3L$, by solving the above equation. After drawing a large enough
(constant) number of points from the distribution, we can assume the expected value is reached in any instance
of the problem. If we rewrite $\alpha = (2a+b)/(3a+b)$ that we computed earlier in the proof of lemma~\ref{lemma:candidate2}, in terms of $L$ and $c$, and replace $c = 0.3L$, we get $a \approx 0.35L$, $b \approx 0.65L$ and $\alpha = 0.79$. Therefore, we can bound expected running time by $(1/(1-0.8)) \times n = 5n$. We need to do an initial scanning of the list in the first candidate elimination and read $n$
points. Therefore, we compute $6n$ as an upper bound for the expected number of points read during
the execution.
\end{proof}

Regardless of the exact running time, by proving Theorem~\ref{theo:linearProof}, we have shown the expected running
time of the algorithm is $O(n)$. In the next section, in our experiments we use counters for the number
of points read until we find the convex hull for each experiment. In all of our experiments,
we read almost $4n$ points during execution. This emphasises, the importance and reliability of our theoretical
analysis for computing the expected running time in this section.

\begin{mytheo}
Expected running time is linear if $P$ follows a \emph{Normal} distribution.
\end{mytheo}
\begin{proof}
We have already done the proof for \emph{Uniform} distribution. We know Normal
distribution is more centered around its mean and further from its boundaries.
It is obvious that this results in more probability mass in eliminated areas
in all of the proofs regarding the expected running time analysis. We can say the
expected running time when $P$ is Uniform is an upper bound for the expected
running time when $P$ is Normal.
\end{proof}

\subsection{ConvexHull Running Time: Experimental Analysis}
\label{sec:Convexpn}

We performed $6$ experiments for different number of points in $P$. The number
of points grows exponentially. In all cases, we generate the point set randomly
from uniform distribution. We use $ClassicAlg$ for the classic algorithm and
$NewAlg$ for our new algorithm based on candidate elimination. We use $QuickSort$ with expected $O(nlog(n))$ time
for sorting in the implementation of classic algorithm which is typically one of the most efficient in practice. In all cases, except
for $n = 10$, our running time is almost $4n$. In the case of $n=10$, we perform
$68$ point reads which is more than $40$. Although the difference is negligible, we relate
the additional cost paid to the overhead of finding four quarters of the convex hull separately.
There exists a negligible amount of overhead because $X_{min}$, $X_{max}$, $Y_{min}$ and $Y_{max}$
belong to candidate sets in more than $1$ quarters.

\begin{center}
\scriptsize
     \begin{tabular}{| l | l | l | l | l | l | l |}
     \hline
     \textbf{$\# Points$} & $10$ & $10^2$ & $10^3$ & $10^4$ & $10^5$ & $10^6$ \\ \hline
     \textbf{$ClassicAlg.$} & $145$ & $2367$ & $32425$ & $474853$ & $14462151$ & ? \\ \hline
     \textbf{$NewAlg.$}& $68$ & $424$ & $4237$ & $39854$ & $398406$ & $3879651$ \\ \hline
     \end{tabular}
 \end{center}

The remarkable results we observe in the above table are: 1) Linear number of
point reads compared to the input size for our new algorithm; 2) Finding $2-$D convex hull
of up to $10^6$ points while the classic algorithm fails to do so. We also notice we pay
a lot less cost in order to find the convex hull of $10^6$ points than the classic algorithm
pays to find the convex hull of $10^5$ points.

\section{Classification Algorithm}
\label{sec:classificationAlgo}
\begin{figure}[t]
  \centering
  \includegraphics[scale=0.4, trim = 20mm 0mm 0mm 0mm]{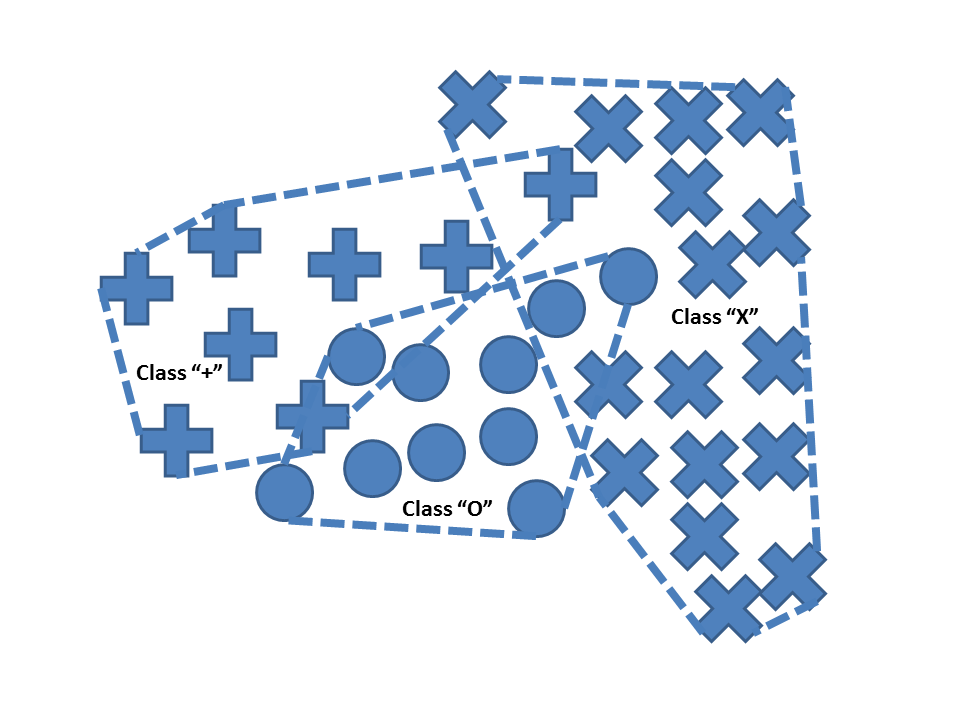}
  \caption{Convex Hull used as a boundary for each class. Convex Hulls may partially overlap.}
  \label{fig:classificationexp1}
\end{figure}
Figure~\ref{fig:classificationexp1}, shows the same distribution of data
in classes as Figure~\ref{fig:classificationexp}. It also shows how
classes are separated using their 2-dimensional (2D) convex hulls.
Any new sample with missing label is checked against these three convex hulls.
It is classified in that class if it is inside the corresponding convex hull.
As shown in the figure, these classes overlap in the areas they cover and misclassification
is always possible. It is also the case that if we try to separate these classes
using other decision boundaries we face the same problem. Since convex hull tightly wraps
around the points from each class, it reduces the chances of misclassification
using its tight boundaries. There can be instances where the point falls outside
all convex hulls. In such cases, we can assign a point to a class using it's
proximity in Euclidean space. In the rest of this paper, we only deal with classification
problems where there are typically $d > 2$ features. In this case, since there are
more than 2 dimensions, for each class we produce 2D convex hulls for every permutation
of 2 features resulting in $\binom{d}{2}$ convex hulls for $d$ features. We define
the notion of \emph{2DAspect} as follows:

\begin{mydef}
A two dimensional data aspect \\
(\emph{2DAspect}) is a structure containing the following
data.
\begin{itemize}
  \item \textbf{ClassLabel} (String): the class this 2DAspect belongs to.
  \item \textbf{$f_1$, $f_2$} (int): indices of a pair of features chosen from the set of all possible pairs.
  \item \textbf{UpperHull}: ordered \emph{list} of points that form the upper hull of data in $f_1$,$f_2$ plane.
  \item \textbf{LowerHull}: ordered \emph{list} of points that form the lower hull of data in $f_1$,$f_2$ plane.
\end{itemize}
\end{mydef}

In order to check whether a point is covered by a 2DAspect, we check if it is below all the lines
on the UpperHull and above all the lines on the LowerHull. Our convex hull based classifier, is composed of
$C \times \binom{d}{2}$ 2DAspect's, while $C$ is the number of classes.
Algorithm~\ref{algo:train}, provides the pseudo-code for training a DataGrinder using 2D convex hulls.

\begin{algorithm}
\begin{algorithmic}[1]
\caption{TrainDataGrinder$(X,Y)$}
\label{algo:train}
\REQUIRE $X_{n \times d}$ data matrix, $Y_{n \times 1}$ corresponding
labels of rows in $X$.
\ENSURE All $C \times \binom{d}{2}$ two dimensional aspects
\STATE $2DAspects = $ empty list
\FOR{\textbf{\emph{each}} class $C_i$ }
\FOR{\textbf{\emph{each}} pair of features $(f_1, f_2)$}
\STATE $P=\pi_{f_1,f_2}(\sigma_{Y=C_i}((X,Y))) $
\STATE $UH(P) = $ upper hull of $P$
\STATE $LH(P) = $ lower hull of $P$
\STATE $2DAspect$ temp = \textbf{new} $2DAspect()$
\STATE $temp.classLabel = C_i$
\STATE $temp.UpperHull = UH(P)$
\STATE $temp.LowerHull = LH(P)$
\STATE $temp.f_1 = f_1$
\STATE $temp.f_2 = f_2$
\STATE $2DAspects.add(temp)$
\ENDFOR
\ENDFOR
\STATE \textbf{return} $2DAspects$
\end{algorithmic}
\end{algorithm}

In Algorithm~\ref{algo:train}, for each class label ($C_i$), and pair of columns (features $f_1$, $f_2$),
we select all rows of $X$ corresponding to $C_i$, then project to columns $f_1$ and $f_2$. Both selection
and projection are standard Relational Algebraic operations and thus we can even implement DataGrinder
inside a database engine. We find upper and lower hulls of the point set, $P$, in the $(f_1,f_2)$ plane.
Having found the convex hull, we create a new $2DAspect$ structure using $C_i$, $f_1$, $f_2$, $UH$ and $LH$.
We repeat the process and construct all $C \times \binom{d}{2}$ 2DAspects.

Testing for a new sample without label is done as follows:

\begin{itemize}
  \item Iterate over all $2DAspects$.
  \item Project the input $x'$ vector to the corresponding $(f_1, f_2)$ for each 2DAspect.
  \item Check if the 2DAspect contains $\pi_{f_1,f_2}(x')$.
  \item Increment the score for the corresponding class label $C_i$.
  \item Find the class $c_{max}$ with the highest score and classify $x'$ to $c_{max}$.
\end{itemize}

Testing is simpler than training and all we need to do is check for all 2DAspects, if they
contain the new data row $x'$, inside their convex hull. Having done this, we keep track of
a \textbf{\emph{count}} for each class, $C_i$, in how many 2DAspects it covers $x'$.
We choose the class with the highest score and assign the appropriate class label according to
DataGrinder.

\begin{figure*}[t]
\centering
\subfigure[$\lambda=1$]{
\includegraphics[width=\figthree]{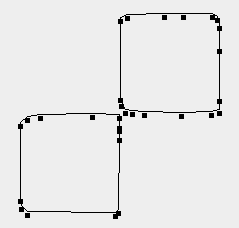}
\label{fig:expr1}
}
\subfigure[$\lambda=2$]{
\includegraphics[width=\figthree]{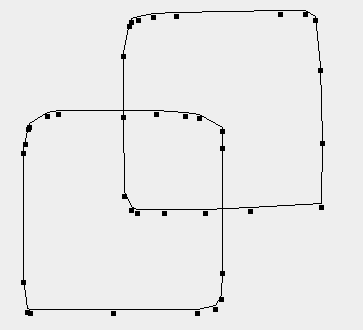}
\label{fig:expr2}
}
\subfigure[$\lambda=3$]{
\includegraphics[width=\figthree]{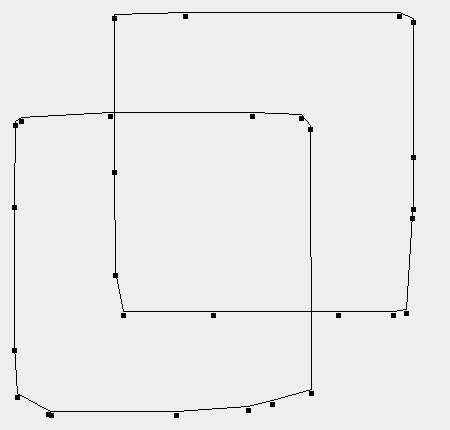}
\label{fig:expr3}
}
\subfigure[$\lambda=4$]{
\includegraphics[width=\figthree]{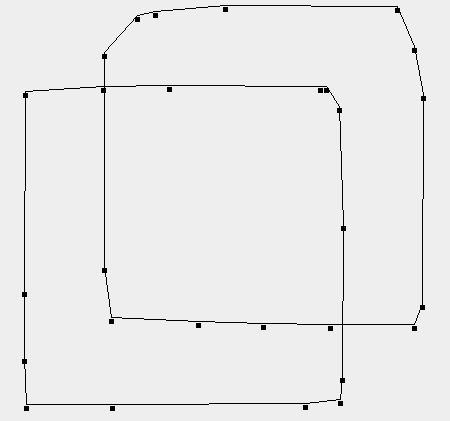}
\label{fig:expr4}
}
\subfigure[$\lambda=5$]{
\includegraphics[width=\figthree]{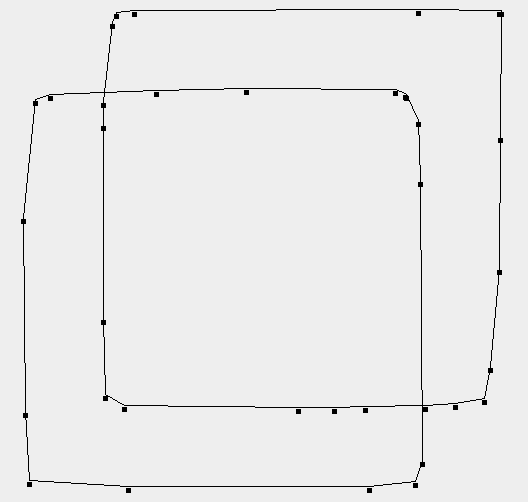}
\label{fig:expr5}
}
\subfigure[$\lambda=10$]{
\includegraphics[width=\figthree]{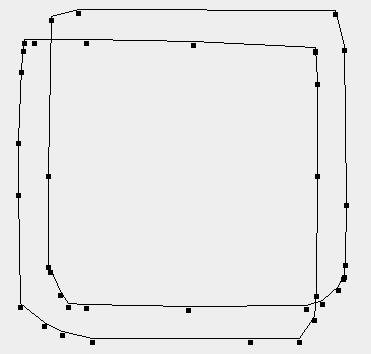}
\label{fig:expr110}
}
\caption{Example: Randomly generated 2D data with different values of $\lambda$ and $C=2$ (binary classification).
Class boundaries are shown after connecting the points on convex hull for each class. While still the two classes
are separable according to their boundaries, it gets harder to distinguish between the classes for larger $\lambda$
due to overlapping areas.}
\label{fig:randomData}
\end{figure*}

\subsection{Parallelization for Training and Testing}
We can achieve parallelization for both training and testing phases
easily by partitioning according to either classes or 2DAspects. This
can be done in a straight forward way following a divide and conquer approach.
For instance we can partition data into different classes or partition according
to indices of $(f_1,f_2)$ combinations. Since this is trivial, we only describe
a simple divide and conquer algorithm for finding the 2D convex hull of a point set
$P$, to conclude this section. It is worth to highlight that in Section~\ref{sec:expectedTime},
we already showed both theoretically and empirically that our convex hull algorithm reads
only $O(n)$ expected number of points during its execution. Parallelization of the same
algorithm using divide and conquer strategy obviously does not increase the running time.
In fact, there may be no reason for parallelization in many scenarios. In cases where we want
to build classifiers on demand for millions of points, it is practical to use parallelization.
Algorithm~\ref{algo:parallelConvexHull} provides the pseudo code.

\begin{algorithm}
\begin{algorithmic}[1]
\caption{DivideConquerConvexHull$(P)$}
\label{algo:parallelConvexHull}
\REQUIRE Point Set $P$
\ENSURE Convex Hull of $P$, $CH(P)$
\STATE Partition $P$ into $k$ partitions $\{P_1...P_k\}$
\STATE $P'$ = empty set of points
\FOR{\textbf{\emph{each}} partition $P_i$}
\STATE $CH(P_i) =$ convex hull of $P_i$
\STATE add all the points in $CH(P_i)$ to $P'$
\ENDFOR
\STATE $CH(P) = CH(P')$
\STATE return $CH(P)$
\end{algorithmic}
\end{algorithm}

We use $CH(P)$ to denote the convex hull of $P$. As described earlier, it is composed
of to halves or four quarters each of which is an ordered set of points by $x$, ($f_1$), coordinate.
The idea is simple, first we partition $P$, until the size of each $P_i$ is small enough.
Typical running times can be estimated according to a simple cost-based analysis, and
computing power/trafic available. We find the convex hull of each $P_i$, resulting in
only a few remaining points on $CH(P_i)$, typically constant. Having done this, we merge all
$CH(P_i)$'s. It is guaranteed that we end up with a \textbf{\emph{super set}} of the points required for the
correct answer of $CH(P)$. We find the convex hull of $P'$ trivially in a final step.

\begin{mytheo}
\label{theo:convexdivide}
Algorithm~\ref{algo:parallelConvexHull}, correctly finds the convex hull of $P$, using
divide and conquer.
\end{mytheo}
\begin{proof}
Proof is already explained since $$CH(P) \subseteq \bigcup_{i=1}^{k}CH(P_i)$$.
\end{proof}

\begin{figure*}[t]
\subfigure[Binary classification and changing $\lambda$]{
\includegraphics[scale=0.65]{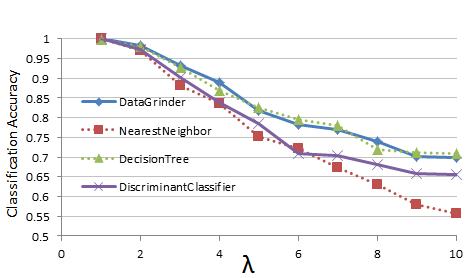}
\label{fig:exLambda}
}
\subfigure[Multi-class classification: fixed $\lambda$ and changing $C$]{
\includegraphics[scale=0.7]{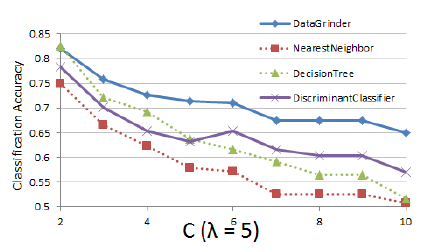}
\label{fig:exC}
}
\caption{Classification accuracy on randomly generated datasets using $C_i + \lambda \times Uniform(0,1)$.}
\end{figure*}

\section{Experimental Analysis}
\label{sec:expn}
We have already shown how our convex hull algorithm achieves expected $O(n)$ point
reads. In this section, we already report our results regarding accuracy in different cases.
First, we propose a random class generation approach, through which we can control the difficulty
of the classification problem instance. We generate data only using the Uniform distribution. It is
known that we can convert other distributions to Uniform as well before classification, using
\emph{Normalization}~\cite{BishopML}.
Here, our focus is mainly on designing DataGrinder and efficient algorithms.
We report our raw results using only the algorithms described and avoid any pre/post processing
to leave more room for the future work, and study the \emph{key} factors involved
in classification accuracy of DataGrinder, in its standard and straight forward case.
We will show shortly, how DataGrinder (DGR) achieves high accuracy even in its simplest form, as
described in this paper. This increases our hopes for designing highly scalable Data Mining and
Machine Learning algorithms, in the Database community. Our data generation aapproach works as follows.
Feature values of class $C$ (last class label), are
generated as: $C + \lambda \times Uniform(0,1)$. This results in producing uniformly distributed
random values for all features of class $C_i$, in the range $[C_i, C_i + \lambda]$. For simplicity, we assume all
class labels are integer, and all features are generated from the same distribution.
Suppose there are only two classes, Figure~\ref{fig:randomData}, shows how
the classification gets more complicated with increasing $\lambda$. There are two class labels $0$ and $1$.
In the case of $\lambda = 1$, the classes are linearly separable from each other. Therefore, any algorithm must
be able to achieve $100\%$ classification accuracy, if both training and testing datasets are generated
using the same $\lambda$ and $C$ parameters. As we increase $\lambda$, the two classes overlap in larger regions
and thus the classification gets more complicated. We have shown our decision boundaries using convex hulls
and points on them, for different values of $\lambda$ (~\ref{fig:randomData}).
It is also commonly known that when there are
more classes (i.e. multi-class classification), the classification is more challenging. This is because
we need more decision boundaries, and there is more probability for overlapping areas as well as fewer
training samples for each class, compared to the number of samples from other classes. Here, we only
show examples of 2D convex hulls for binary classification. As described earlier there are $2 \times \binom{d}{2}$
such \textbf{\emph{"2DAspects"}}. In each case, we generate data ($X$) with 5 dimensions $f_1,...,f_5$. Figure~\ref{fig:exLambda},
compares DataGrinder classification accuracy, to $3$ other well-known methods for $C=2$, and changing $\lambda$. For all the other three algorithms, we use standard
MATLAB functions and default parameter setting, since DataGrinder is fully non-Parametric.
DecisinTree, is a text-book classifier,
that achieves optimization using partitioning, information gain and obtaining a sequence
of comparisons that leads to a class label with high accuracy. NearestNeighbor method searches the training
dataset for a new testing instance, and assigns class label according to the closest point in
the Euclidean space. DiscriminantClassifier, finds decision boundaries using $L_1$ and $L_2$
Regularization~\cite{DiscriminantClassifier:07}, in order to avoid overfitting to the training data.
We train and test using
$1000$ samples for training and testing each. Both Decision Tree and DiscriminantClassifier
may need heavy training time if the dataset is large due to their optimization problems. Typically, at least
several \emph{Sequential Scans} of the dataset is the minimum required.
NearestNeighbor is the most efficient, if we use space partitioning spacial indices in testing.
However, the results show its accuracy is outperformed by all methods
almost in all cases. Using our randomly generated data for binary classification, we find that
Decision Tree and DataGrinder achieve the highest accuracy. We believe this is due to the fact
that they both partition the space into regions rather than just using lines or hyperplanes as
decision boundaries and our classification scenario is such that the DiscriminantClassifier fails.
We fix $\lambda=5$ and repeat for multi-class classification while changing $C$. In this case, we notice
all classification algorithms fail compared with DataGrinder, due to the considerable gap in accuracy.
Given DataGrinder's special scalability features for BigData, this is a bonus that DataGrinder also
achieves outstanding accuracy in this experiment compared with commercial classification algorithms in MATLAB 2012.

\begin{figure}[t]
  \centering
  \includegraphics[scale=0.8, trim = 20mm 0mm 0mm 0mm]{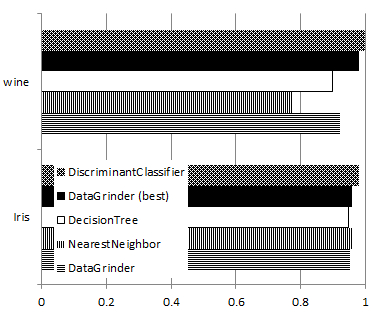}
  \caption{Classification accuracy on Wine and Iris datasets.}
  \label{fig:classificationexpIrisWine}
\end{figure}

\begin{figure}[t]
  \centering
\includegraphics[scale=0.5, trim = 30mm 0mm 20mm 0mm]{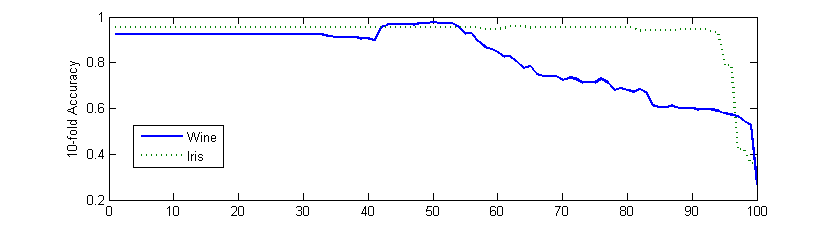}
\caption{Using 100 "filters", for $\theta = 0,0.01,0.02...1$ to find the best classifier.}
\label{fig:RatioExp}
\end{figure}

\subsection{Existing Classification Datasets}
We use two standard datasets also used as examples in Machine Learning textbooks for
the classification problem, Iris and Wine. We obtain these datasets from the UCI data mining repository~\footnote{http://archive.ics.uci.edu/ml/datasets.html}.
Both datasets have less than $1000$ samples
and $3$ classes. We use 10-fold cross validation for training and testing, meaning
we divide the dataset into $10$ partitions, and use the average of $10$ experiments.
In each experiment, we use $9$ partitions for training and $1$ for testing. All algorithms
reach acceptable accuracy on Iris dataset $>90\%$, and close to $1$ (Figure~\ref{fig:classificationexpIrisWine}).
In the case of Wine dataset, DiscriminantClassifier performs slightly superior compared to DataGrinder and DecisionTree. NearestNeighbor
method is significantly outperformed by all the other algorithms. Since DiscriminantClassifier uses many
parameters to achieve this, we also decide to add only $1$ hyper-parameter namely \textbf{\emph{Filtering Ratio ($0\le \theta<1$)}} to DataGrinder. In order to do this, we add an additional variable to each 2DAspect, \emph{Classification Accuracy}. It refers to the number of training samples that correctly fall inside a 2DAspect (i.e. 2DAspect and
data labels match), over the total number of all samples. Any 2DAspects that largely overlap with other classes
resulting in classification accuracy less than $\theta$, are removed from DataGrinder. We vary $\theta$ using
a $0.01$ step size from $0$ to $1$ over the training dataset and record the best testing accuracy. We also
show in Figure~\ref{fig:classificationexpIrisWine}, the best DataGrinder accuracy after filtering using
a solid bar. As it is notable, DataGrinder accuracy increases after filtering, resulting in less classification
errors. This also adds another dimension to our future research in order to target adding few meaningful
parameters or hyper-parameter to the model that increase accuracy. The remarkable fact to highlight about
the filtering technique presented is that we can achieve higher accuracy using tuning techniques and
this leaves the door open for future research on DataGrinder. Figure~\ref{fig:RatioExp}, shows how
DataGrinder accuracy changes on these datasets with varying $\theta$. For $\theta=0$, i.e. Raw DataGrinder, no
2DAspects are filtered. When $\theta=1$, all 2DAspects are filtered and all samples are assigned to the default class $0$. In both cases we get the best accuracy around $\theta = 0.5$. \textbf{\emph{This is logical, because any features whose
classification rate is more than misclassification rate can be useful for discriminating between classes. When
a large enough number of such features are combined, we can achieve high overall accuracy.}}

\section{Related work}
\label{sec:related}
In this section, we review the recent works in literature that
discuss scalable data mining algorithms and frameworks similar
to DataGrinder, in motivation and technical contribution.

In~\cite{SVMInd:11}, authors propose an "Exact Indexing" approach
for Support Vector Machines. They propose indexing strategies
in the Kernel space, iKernel, that is used for exact top-k query processing
when SVM is used for ranking. Given a SVM model, authors use properties
of the Kernel space such as \emph{ranking instability} and \emph{ordering stability}.
They provide an excellent background of support vector machines, and their relevance to
databases, top-k query processing and ranking. They only focus on prediction (i.e. testing),
and do not aim at designing parallel SVM algorithms. DataGrinder, provides a highly scalable
algorithmic framework for both training, model updating and testing that achieves high accuracy.
We might be able to focus on future work leveraging convex hulls for constructing kernels as well
as ranking. Although this requires leveraging more geometric properties of the data, in order
to be able to achieve accuracy as high as Support Vector Machines. Support
Vector Machine is a well researched problem with a complex structure. In contrary, DataGrinder
aims at building simpler discrete models with high accuracy and our initial experimental
results are promising. SVM also has applications in bioinformatics, where there are
thousands of features and we need to improve DataGrinder in order to be able to deal with
these applications. Biological datasets are typically more complex. Regardless of the model
structure, they focus on ranking and top-k query processing while we focus on convex hulls
and classification. ArrayStore~\cite{ArrayStore:11}, is a storage manager for complex array processing. Authors
process datasets as big as $80GB$, using parallel data mining algorithms. They provide a multi-dimensional
array model, suitable for our classification scenario. They also discuss data access issues. Our
\emph{Select}-\emph{Project}-\emph{ConvexHull} series of operations completely fits within their
storage framework. Thus, we do not worry about scalability of DataGrinder at all. Rather than focus
on storage, in this paper we discuss a \emph{new discrete classification algorithm}, that can work on the
top of ArrayStore. Authors already discuss two types of clustering algorithms, but they did not provide
any examples on the classic classification problem. We provide a divide and conquer algorithm that makes
DataGrinder compatible with ArrayStore. ERACER~\cite{ERACER:10}, provides an iterative statistical framework for filling in missing data, as well
as data cleaning and fixing corrupted values using conventional statistical methods. DataGrinder can solve their problem in a special case. Extensions of DataGrinder can also
solve the same exact problem. We use Computational Geometry, and theoretical analysis for a $O(n)$ expected
running time algorithm while maintaining accuracy. DataGrinder can as well fit inside a DBMS engine using
\emph{Select}-\emph{Project}-\emph{ConvexHull}. We can implement ConvexHull as an operation using \emph{Table Functions}. DataGrinder is fully non-Parametric, meaning that it is easy to use, and needs no parameter tuning.
DataGrinder is completely discrete and we can also count on divide and conquer solutions for intense scalability.
DataGrinder is easy to implement, thus suitable for the industry. DataGrinder achieves high accuracy in classification. We also show with experiments how we can improve accuracy by \emph{Filtering($\theta$)}. All in all, we find DataGrinder a more suitable solution for the database community, due to its strong and fundamental
theoretical contributions. Spanners~\cite{Spanners:13}, is an interesting theoretical contribution, and a formal framework
for information extraction. We believe DataGrinder has a similar flavour in its contribution to Spanners.
We also aim at designing operations for processing multidimensional data and knowledge discovery. Spanners is focused
on Information Extraction and using Regular Expressions for Text Mining using predefined operations. Several other
previous works have also tried to achieve the same goal such as~\cite{TopRecs:11}. Another interesting direction to achieve parallel statistical and data mining algorithms is through Sampling~\cite{DCInference:12}. In this approach, we make \emph{BigData} assumption and use parallelization
for processing. We build many small models and using statistical inference, we combine these models to
guarantee reliability and accuracy. The size of input data and distribution(s) of data  are examples
of key parameters we need to take into account. Naturally, we need to focus on how to sample and
pay attention to things such as the number of samples, the size of each sample as well as
how to effectively combine the models built using different samplings of the data.
This can be done for Big Data, regardless of the data mining task discussed. Examples of such methods include~\cite{ClusterForest:13, DDFactors:11}. Rule-based classifiers are other examples of discrete classification algorithms, discussed
in the data mining literature~\cite{RuleClass1, RuleClass2, RuleClass3}. They use frequent
patterns and association rules mining in order to find rules with high support and confidence.
They typically achieve reliable accuracy. They need a rather costly parameter tuning step to construct the best
classifier. They need the exact solution of a NP-hard theoretical problem compared to $O(n)$ expected
running time of DataGrinder. There have been some attempts for parallel frequent pattern mining
algorithms which is \emph{outside the context of this paper}.

Convex Hull problem has a long history in Computational Geometry~\cite{ComputationalGeometryBook}.
It is significantly important, because many other important problems in Computational Geometry
can be reduced to this problem. Many efforts have been devoted to improving the \emph{worst case}
running time and output sensitive algorithms.
We find average case analysis more suitable to the database community, due to its
similarity to cost-based query optimization. In~\cite{Convex:81}, there is a proposal for
expected $O(n)$ algorithms along with theoretical analysis to prove its possibility. In this paper,
we provide an \emph{algorithm with pseudo code} and calculate a \emph{exact costant}, to serve
as an upper-bound for the expected running time. Our experimental evaluation backs up all of our
arguments, regardless of the running time and programming languages used.

\section{Conclusions and Future Work}
\label{sec:concl}
In this paper, we revisited the important problem of finding 2D convex hulls. We
propose an algorithm based on a well-known historical algorithm, with $O(n)$ expected running time.
We propose a simpler and shorter proof compared to the previous work, and also calculate a
constant that serves as an upper-bound for the expected linear running time. We conduct experiments
to back up all of our arguments. We perform several experiments and show DataGrinder is comparable
to the most reliable commercial classification packages in MATLAB and outperforms many, while maintaining
its extreme provable scalability. We show how to achieve several
levels of parallelization, while keeping the correctness of our classification algorithm. We intend to focus
on more detailed Geometrical study of the problem, in order to partition the data more accurately. Specially,
remove sparse areas. We also intend to test for more classification scenarios as well as adding meaningful parameters and hyper-parameters to DataGrinder. We would also like to take DataGrinder to the cloud for classification of enormous datasets.

\scriptsize{

}
\end{document}